\documentclass[aps,pra,reprint,superscriptaddress,amsmath,amssymb,floatfix]{revtex4-2}

\usepackage{graphicx}
\usepackage{bm}
\usepackage{amsthm}
\newtheorem{theorem}{Theorem}
\newtheorem{lemma}[theorem]{Lemma}
\usepackage[colorlinks=true,allcolors=blue]{hyperref}

\newcommand{\Ic}{I_{c}}
\newcommand{\id}{\mathrm{id}}
\newcommand{\eig}{\mathrm{eig}}

\begin{document}

\title{Dephasing channels have no entanglement-breaking threshold:\\
a caution for threshold claims in quantum biology}

\author{Hikaru Wakaura}
\email{hikaruwakaura@gmail.com}
\affiliation{QIRI (Quantum Integrated Research Institute Inc.),
1--16--3 Akasaka, Minato-ku, Tokyo 107--0061, Japan}

\author{Taiki Tanimae}
\affiliation{QIRI (Quantum Integrated Research Institute Inc.),
1--16--3 Akasaka, Minato-ku, Tokyo 107--0061, Japan}

\date{\today}

\begin{abstract}
A recurring argument in quantum biology is that decoherence above some critical
rate renders a biological channel entanglement-breaking, and therefore useless
as a quantum resource. We show that for the two channel families in which this
argument is usually posed no such critical rate exists. For uniform dephasing at
rate $\gamma$ on a $d$-dimensional system, the partial transpose of the Choi
state has eigenvalue $-e^{-\gamma}/d$ on every antisymmetric vector, so the
channel is non-PPT---hence not entanglement-breaking---at every finite
$\gamma$. Adding amplitude damping does not change this: that channel is a
tensor power of a qubit map whose partial transpose has minimum eigenvalue
$-c^{2}/(p+\sqrt{p^{2}+4c^{2}})<0$ for all finite rates. What does possess a
finite threshold is the single-letter coherent information, which vanishes near
$\gamma\simeq0.62$. We further identify the numerical mechanism that
manufactures the missing threshold: because the minimum partial-transpose
eigenvalue approaches zero exponentially without reaching it, a PPT test with a
fixed absolute tolerance $\epsilon$ reports a threshold at whatever $\gamma$
satisfies $\min\eig\,C^{T_{A}}=-\epsilon$; $\epsilon=10^{-10}$ yields
$\gamma=5.49$ and $\epsilon=10^{-12}$ yields $6.58$. We encountered this
artifact in our own now-retracted work, and give the closed forms and code that
avoid it.
\end{abstract}

\maketitle

\section{Introduction}

Whether coherent spin dynamics can survive in warm, wet biological matter is an
old question~\cite{Tegmark2000} that has become quantitative for the radical-pair
mechanism of magnetoreception~\cite{RitzAdemSchulten2000,HoreMouritsen2016,
Kominis2015} and for proposals of nuclear-spin processing in
neurons~\cite{Fisher2015}. A standard way to pose it is to ask how much
entanglement, or how much quantum information, survives a given decoherence
rate~\cite{CaiPlenio2010,GaugerPlenio2011}.

A particular version of this question has proved seductive: one models the
biological environment as a decoherence channel $\mathcal{N}_{\gamma}$, asserts
that above some critical rate $\gamma_{c}$ the channel becomes
\emph{entanglement-breaking} (EB), and concludes that any quantum advantage is
extinguished beyond $\gamma_{c}$---so that the biologically interesting question
is whether the system sits above or below it. The appeal is obvious: EB is a
sharp, binary, physically meaningful property~\cite{HorodeckiShorRuskai2003}, and
a threshold invites a clean verdict.

We show here that for the channel families in which this argument is actually
posed, the threshold does not exist. Both are non-PPT---and therefore not
entanglement-breaking---at \emph{every} finite decoherence rate; they become EB
only in the limit $\gamma\to\infty$. Any reported finite $\gamma_{c}$ for these
channels is therefore either a different quantity under an incorrect name, or a
numerical artifact.

Three inequivalent properties are routinely conflated in this literature, and
separating them is most of the content of this paper:
\begin{enumerate}\itemsep2pt
\item[(i)] $\mathcal{N}$ is entanglement-breaking, i.e.\ its Choi state is
      separable;
\item[(ii)] the single-letter coherent information vanishes,
      $\max_{\rho}\Ic(\rho,\mathcal{N})=0$;
\item[(iii)] the quantum capacity vanishes, $Q(\mathcal{N})=0$.
\end{enumerate}
These satisfy (i)\,$\Rightarrow$\,(iii)\,$\Rightarrow$\,(ii) but no converse.
Property (ii) does not imply (iii), because the coherent information is
superadditive: there are channels with $\max_{\rho}\Ic=0$ and $Q>0$
\cite{DiVincenzo1998,SmithYard2008}, and no finite number of channel uses
suffices in general to certify $Q=0$~\cite{Cubitt2015}. Property (iii) does not
imply (i), the PPT entangled channels being the standard
counterexample~\cite{HorodeckiBound1998}. For the channels considered here only
(ii) has a finite threshold, and it is (ii) that a numerical search of the usual
kind will find.

Finally we identify the specific numerical trap that converts ``no threshold''
into a confident number. It is not exotic: it is a fixed absolute tolerance in a
PPT test applied to a quantity that decays exponentially to zero. We fell into
it ourselves. Three preprints of ours built on a threshold obtained this way and
have been retracted~\cite{P1retraction,P2retraction,P3retraction}; the present
paper is the general statement of what went wrong, addressed to the channel
theory rather than to the biology.

\section{Channels and criteria}

Write $\mathcal{N}$ for a quantum channel on a $d$-dimensional system and
\begin{equation}
C_{\mathcal{N}}=(\id\otimes\mathcal{N})\bigl(|\Phi^{+}\rangle\langle\Phi^{+}|\bigr),
\qquad
|\Phi^{+}\rangle=\frac{1}{\sqrt d}\sum_{i=0}^{d-1}|ii\rangle ,
\end{equation}
for its normalised Choi state. The Choi--Jamio{\l}kowski isomorphism makes the
following exact~\cite{HorodeckiShorRuskai2003}: $\mathcal{N}$ is
entanglement-breaking if and only if $C_{\mathcal{N}}$ is separable. We test
this through the Peres--Horodecki criterion~\cite{Peres1996,HorodeckiPLA1996}: if
the partial transpose $C^{T_{A}}$ has a negative eigenvalue then
$C_{\mathcal{N}}$ is entangled and $\mathcal{N}$ is \emph{not} EB. Negativity of
the partial transpose is a certificate of non-EB; positivity alone is not a
certificate of EB when $d^{2}>6$, but we will not need that direction.

The coherent information of a state $\rho$ through $\mathcal{N}$ is
\begin{equation}
\Ic(\rho,\mathcal{N})
 = S\bigl(\mathcal{N}(\rho)\bigr)
 - S\bigl((\id\otimes\mathcal{N})(\psi_{\rho})\bigr),
\label{eq:ic}
\end{equation}
with $\psi_{\rho}$ any purification of $\rho$ and $S$ the von Neumann
entropy~\cite{SchumacherNielsen1996}. By the LSD
theorem~\cite{Lloyd1997,Devetak2005,DevetakShor2005} the quantum capacity is the
regularisation $Q=\lim_{n}\frac1n\max_{\rho}\Ic(\rho,\mathcal{N}^{\otimes n})$,
of which $\max_{\rho}\Ic(\rho,\mathcal{N})$ is only a lower bound.

We consider two families.

\paragraph*{(A) Uniform dephasing.} Every off-diagonal matrix element is
multiplied by $\lambda=e^{-\gamma}$,
\begin{equation}
\mathcal{D}_{\gamma}=\lambda\,\id+(1-\lambda)\,\Delta ,
\qquad \lambda=e^{-\gamma},
\label{eq:dephasing}
\end{equation}
where $\Delta(\rho)=\sum_{i}|i\rangle\langle i|\rho|i\rangle\langle i|$ is
complete dephasing in the computational basis. This is the channel written down
in the models we are addressing.

\paragraph*{(B) Dephasing with amplitude damping.} The generator that such
models actually integrate adds relaxation: each of $n$ qubits evolves under
Lindblad operators $\sqrt{\gamma}\,\sigma_{z}$ and
$\sqrt{\kappa\gamma}\,\sigma_{-}$ for a time $t$, with $\kappa$ the ratio of
relaxation to dephasing (we take $\kappa=0.1$, $t=1$, $n=2$ unless stated).
Integrating gives the qubit map
\begin{equation}
\rho_{11}\mapsto(1-p)\rho_{11},\qquad
\rho_{01}\mapsto c\,\rho_{01},
\label{eq:qubitmap}
\end{equation}
with damping probability and coherence factor
\begin{equation}
p=1-e^{-\kappa\gamma t},
\qquad
c=e^{-2\gamma t}\sqrt{1-p},
\qquad
c^{2}=e^{-(4+\kappa)\gamma t}.
\label{eq:pc}
\end{equation}

\section{Uniform dephasing is never entanglement-breaking}

\begin{theorem}\label{thm:dephasing}
For every finite $\gamma$ and every $d\ge2$, the Choi state of
$\mathcal{D}_{\gamma}$ satisfies
\begin{equation}
\min\eig\,C^{T_{A}}_{\mathcal{D}_{\gamma}}=-\frac{e^{-\gamma}}{d}<0 .
\label{eq:thm1}
\end{equation}
Hence $\mathcal{D}_{\gamma}$ is non-PPT and not entanglement-breaking at any
finite rate. It is EB only in the limit $\gamma\to\infty$.
\end{theorem}

\begin{proof}
Applying $\id\otimes\mathcal{D}_{\gamma}$ to $|\Phi^{+}\rangle\langle\Phi^{+}|$
multiplies the $|ii\rangle\langle jj|$ entries with $i\neq j$ by $\lambda$ and
leaves the diagonal untouched, so
\begin{equation}
C=\lambda\,|\Phi^{+}\rangle\langle\Phi^{+}|
 +\frac{1-\lambda}{d}\sum_{i}|ii\rangle\langle ii| .
\label{eq:choi-deph}
\end{equation}
Partial transposition on the first factor sends
$|ii\rangle\langle jj|\mapsto|ji\rangle\langle ij|$ and fixes
$|ii\rangle\langle ii|$, so
\begin{equation}
C^{T_{A}}=\frac{\lambda}{d}\sum_{i\ne j}|ji\rangle\langle ij|
 +\frac{1}{d}\sum_{i}|ii\rangle\langle ii| .
\end{equation}
The first term is $(\lambda/d)$ times the swap operator restricted to the
off-diagonal sector. On the antisymmetric vector
$|\psi^{-}_{ij}\rangle=(|ij\rangle-|ji\rangle)/\sqrt2$ with $i\ne j$ it acts as
multiplication by $-1$, while the second term annihilates
$|\psi^{-}_{ij}\rangle$. Hence
$C^{T_{A}}|\psi^{-}_{ij}\rangle=-(\lambda/d)|\psi^{-}_{ij}\rangle$: every one of
the $d(d-1)/2$ antisymmetric vectors is an eigenvector with eigenvalue
$-\lambda/d$. The remaining eigenvalues are $+\lambda/d$ on the symmetric
off-diagonal vectors and $1/d$ on the $|ii\rangle$, all non-negative, so
$-\lambda/d$ is the minimum. Since $\lambda=e^{-\gamma}>0$ for all finite
$\gamma$, the claim follows.
\end{proof}

Two remarks. First, Eq.~\eqref{eq:thm1} depends on $\gamma$ only through an
overall exponential factor; there is no structural change anywhere along the
family, and in particular nothing that could be read as a transition. Second,
the value is the same for every $d$ up to the $1/d$ prefactor, so no threshold
can be recovered by enlarging the system.

\section{Amplitude damping does not restore the threshold}

One might expect that adding relaxation---which does drive the system towards a
fixed pure state---creates the missing threshold. It does not.

\begin{lemma}\label{lem:tensor}
Let $\mathcal{M}$ act on one qubit and let $\mu_{\min}\le0\le\mu_{\max}$ be the
extreme eigenvalues of $C^{T_{A}}_{\mathcal{M}}$. Then
\begin{equation}
\min\eig\,C^{T_{A}}_{\mathcal{M}^{\otimes n}}
 = \mu_{\min}\,\mu_{\max}^{\,n-1}.
\end{equation}
\end{lemma}

\begin{proof}
Choi states are multiplicative,
$C_{\mathcal{M}^{\otimes n}}=C_{\mathcal{M}}^{\otimes n}$ up to the local
reordering that groups reference and system registers, and partial transposition
acts factorwise, so
$(C^{\otimes n})^{T_{A}}=(C^{T_{A}})^{\otimes n}$. The spectrum of a tensor
power is the set of products of $n$ eigenvalues, and the algebraically smallest
such product pairs the single negative extreme with $n-1$ copies of the largest.
\end{proof}

\begin{theorem}\label{thm:damping}
For the channel of Eqs.~\eqref{eq:qubitmap}--\eqref{eq:pc} on $n$ qubits,
\begin{equation}
\min\eig\,C^{T_{A}}
 = \frac{-c^{2}}{p+\sqrt{p^{2}+4c^{2}}}\left(\frac{1}{2}\right)^{n-1}
 \;\simeq\;-\,\frac{e^{-(4+\kappa)\gamma t}}{2^{n}\,p},
\label{eq:thm2}
\end{equation}
which is strictly negative for every finite $\gamma$. This channel too is
non-PPT, hence not entanglement-breaking, at every finite rate.
\end{theorem}

\begin{proof}
For one qubit, Eq.~\eqref{eq:qubitmap} gives
$\mathcal{M}(|0\rangle\langle0|)=|0\rangle\langle0|$,
$\mathcal{M}(|1\rangle\langle1|)=p|0\rangle\langle0|+(1-p)|1\rangle\langle1|$
and $\mathcal{M}(|0\rangle\langle1|)=c\,|0\rangle\langle1|$, so
\begin{align}
2\,C=\;&|00\rangle\langle00|
 +c\,\bigl(|00\rangle\langle11|+|11\rangle\langle00|\bigr)\nonumber\\
 &+p\,|10\rangle\langle10|+(1-p)\,|11\rangle\langle11| .
\end{align}
Partial transposition moves the coherence into the odd sector,
$|00\rangle\langle11|\mapsto|10\rangle\langle01|$, leaving $C^{T_{A}}$ block
diagonal with an even block $\mathrm{diag}(1,1-p)/2$ and an odd block
\begin{equation}
\frac{1}{2}\begin{pmatrix}0 & c\\ c & p\end{pmatrix}
\quad\text{on}\quad \{|01\rangle,|10\rangle\},
\end{equation}
whose eigenvalues are $\bigl(p\pm\sqrt{p^{2}+4c^{2}}\bigr)/4$. Thus
$\mu_{\min}=\bigl(p-\sqrt{p^{2}+4c^{2}}\bigr)/4
=-c^{2}/\bigl(p+\sqrt{p^{2}+4c^{2}}\bigr)$, the second form obtained by
rationalising, and $\mu_{\max}=1/2$. Since the generator acts identically and
independently on each qubit with no coupling Hamiltonian, the $n$-qubit channel
is $\mathcal{M}^{\otimes n}$ and Lemma~\ref{lem:tensor} gives
Eq.~\eqref{eq:thm2}. Finally $c^{2}=e^{-(4+\kappa)\gamma t}>0$ for all finite
$\gamma$, so $\mu_{\min}<0$.
\end{proof}

The asymptotic form in Eq.~\eqref{eq:thm2} is the crux: the obstruction to
entanglement breaking decays like $e^{-(4+\kappa)\gamma t}$, that is,
exponentially in the decoherence rate, but it never vanishes. This is the
behaviour that a numerical threshold search misreads, and we return to it in
Sec.~\ref{sec:artifact}.

Equation~\eqref{eq:thm2} should be evaluated in the rationalised form. The
difference form $\bigl(p-\sqrt{p^{2}+4c^{2}}\bigr)/4$ suffers catastrophic
cancellation and underflows to $-0.0$ in double precision beyond
$\gamma\approx12$, since $p^{2}+4c^{2}$ then rounds to $p^{2}$---which would
itself be read as a threshold.

\section{What does have a threshold}

\begin{figure*}[t]
\centering
\includegraphics[width=0.92\textwidth]{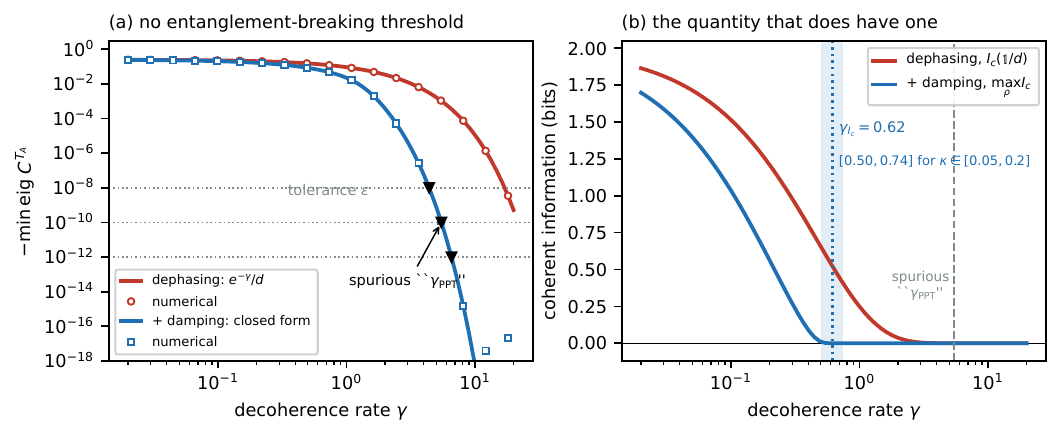}
\caption{\label{fig:main}%
(a) Minimum eigenvalue of the partial transpose of the Choi state, plotted as
$-\min\eig\,C^{T_{A}}$ so that positive values mean non-PPT, hence \emph{not}
entanglement-breaking. Lines are the closed forms of
Eqs.~\eqref{eq:thm1} and~\eqref{eq:thm2}; symbols are computed from the Choi
matrix. Both curves decay exponentially and neither reaches zero: there is no
finite entanglement-breaking threshold. Dotted lines mark absolute tolerances
$\epsilon=10^{-8},10^{-10},10^{-12}$; the triangles mark where a PPT test using
each would declare a threshold ($\gamma=4.41$, $5.49$, $6.58$). The two
right-most numerical symbols lie above the closed form because at
$|\min\eig|<10^{-16}$ the eigenvalue is dominated by solver noise.
(b) Coherent information of the same two channels. This quantity does have a
finite threshold, $\gamma_{\Ic}=0.62$ for $\kappa=0.1$ (shaded band:
$[0.50,0.74]$ over $\kappa\in[0.05,0.2]$), an order of magnitude below the
tolerance-generated value.}
\end{figure*}

We evaluated Eq.~\eqref{eq:ic} on genuine purifications and maximised over
inputs. Both channels are covariant under $e^{i\theta\sigma_{z}}$ on each qubit,
so the maximiser may be taken diagonal: twirling over the covariance group
projects onto the diagonal and cannot decrease the concave functional. We
verified this restriction against an unconstrained search over all $4\times4$
density matrices; the two agree to $10^{-9}$ except within the immediate
neighbourhood of the threshold, where both are below $10^{-4}$ bits.

Figure~\ref{fig:main}(b) shows the result. The single-letter coherent
information of channel (B) vanishes at
\begin{align}
\gamma_{\Ic}&=0.62 \qquad (\kappa=0.1),\nonumber\\
\gamma_{\Ic}&\in[0.50,0.74] \quad \text{for }\ \kappa\in[0.05,0.2].
\end{align}
This is a real threshold, and it is what a numerical study of these channels can
legitimately report. Two cautions attach to it. First, $\max_{\rho}\Ic=0$ does
not imply $Q=0$~\cite{DiVincenzo1998,SmithYard2008,Cubitt2015}, so it is a
statement about one-shot coherent information and not about the channel's
capacity. Second, it is nowhere near the entanglement-breaking property: at
$\gamma=\gamma_{\Ic}$ the Choi state has
$\min\eig\,C^{T_{A}}=-6.3\times10^{-2}$, robustly entangled by any measure.

\section{How a fixed tolerance manufactures a threshold}
\label{sec:artifact}

Implementations test PPT by asking whether $\min\eig\,C^{T_{A}}\ge-\epsilon$ for
some small $\epsilon$ meant to absorb floating-point error. That is sound when
the eigenvalue crosses zero transversally. Here it does not cross at all: by
Eq.~\eqref{eq:thm2} it approaches zero exponentially from below. A bisection on
the predicate therefore converges not to a threshold but to the solution of
$\min\eig\,C^{T_{A}}=-\epsilon$, which is
\begin{equation}
\gamma_{\text{spurious}}
 \simeq\frac{1}{4+\kappa}\,
   \ln\!\frac{1}{2^{n}\,p\,\epsilon}\,,
\end{equation}
a quantity determined by the tolerance. For channel (B) with $\kappa=0.1$,
$n=2$:
\begin{center}
\begin{tabular}{lcccc}
\hline\hline
$\epsilon$ & $10^{-8}$ & $10^{-10}$ & $10^{-12}$ & $10^{-14}$\\
$\gamma_{\text{spurious}}$ & $4.41$ & $5.49$ & $6.58$ & $7.68$\\
\hline\hline
\end{tabular}
\end{center}
The reported ``threshold'' moves by roughly $1.1$ per decade of tolerance. It is
a property of the numerical settings, not of the channel. The value
$\gamma=5.49$ obtained with the common default $\epsilon=10^{-10}$ is precisely
the number we previously reported as an entanglement-breaking threshold.

The diagnostic is simple and we recommend it: vary $\epsilon$ over several
decades. A genuine threshold is stable; an artifact slides logarithmically.

A second, independent numerical hazard affects this same channel family. Built
by exponentiating the Liouvillian, $\mathcal{N}=e^{\mathcal{L}t}$ with
$\mathcal{L}$ the two-qubit Lindbladian, the scaling-and-squaring
routine~\cite{HighamExpm2005,MolerVanLoan2003} in a standard scientific library
overflows during its squaring step for $\gamma\gtrsim3$ and returns a matrix
that is not trace-preserving, corrupting every derived quantity without raising
an error that a caller is likely to check. Because the generator has no
Hamiltonian part, the exact Kraus operators are available in closed form and
should be used instead; they agree with the matrix exponential to $10^{-16}$
throughout the range where the latter remains valid.

\section{Discussion}

The results above are elementary---Theorem~\ref{thm:dephasing} is three lines
once the Choi state is written down---and that is the point. The claim they
refute has nevertheless appeared repeatedly, including in our own work, because
the check that would have caught it was performed numerically with a tolerance
rather than analytically.

For quantum biology the practical consequence is a redirection rather than a
prohibition. Asking whether a biological channel is entanglement-breaking is
asking a question whose answer, for uniform dephasing and for dephasing with
relaxation, is fixed and negative independently of the biology; it cannot
discriminate between organisms, temperatures, or coherence times, because it is
not a threshold phenomenon in these families. Papers reporting a finite
$\gamma_{c}$ for such a channel should re-derive it from Choi separability
directly, and should expect to find either that the quantity computed was
actually the coherent-information threshold, or that it was set by a numerical
tolerance. The quantities that do discriminate---the coherent information, and
below it the operational figures of merit of whatever task the system is
supposed to perform---have thresholds an order of magnitude lower, in a regime
where the biological parameters matter.

We note explicitly that this paper is also a correction. Three preprints by the
present authors asserted an entanglement-breaking threshold at a biological
decoherence rate and built error-correction and reservoir-computing arguments on
top of it. They have been retracted~\cite{P1retraction,P2retraction,
P3retraction}. Theorem~\ref{thm:dephasing} is the reason the premise fails, and
Sec.~\ref{sec:artifact} is the reason we did not notice.

\section*{Code availability}

The closed forms, the Kraus construction, the Choi and partial-transpose
routines, and the coherent-information optimiser are in the
\texttt{channel\_analysis} module of the open \texttt{qbrain3layer} package,
together with the script that produces Fig.~\ref{fig:main}. The module replaces
an earlier routine of ours that returned a state-independent surrogate in place
of Eq.~\eqref{eq:ic} and so could not have located any threshold; that failure
mode is documented in the module for the benefit of anyone who used it.

\begin{acknowledgments}
The re-examination that produced this paper began as an adversarial internal
review of the manuscripts now retracted. We thank the editors and referees who
declined earlier versions of those manuscripts for the objections that prompted
it.
\end{acknowledgments}

\bibliographystyle{apsrev4-2}
\bibliography{refs}

\end{document}